\documentclass[12pt]{article}
\usepackage[UKenglish]{babel} 
\usepackage[T1]{fontenc} 
\usepackage[utf8]{inputenc} 
\usepackage{graphicx} 
\usepackage{amsmath}
\usepackage{amsthm}
\usepackage{amssymb}
\usepackage{cite} 
\usepackage{float} 
\usepackage{bm} 
\usepackage{indentfirst} 
\usepackage{caption} 
\usepackage{subcaption} 
\usepackage[inline]{enumitem} 
\usepackage{comment} 
\usepackage{csquotes}
\usepackage{orcidlink}
\usepackage{hyperref} 
\usepackage{tikz}
\usetikzlibrary{decorations.pathmorphing}
\usetikzlibrary{calc}

\newtheorem{theorem}{Theorem} 

\newtheorem{prop}{Proposition} 

\newtheorem*{assumption*}{Assumption}

\newtheorem*{claim*}{Claim}
\makeatletter
 \DeclareFontFamily{OMX}{MnSymbolE}{}
 \DeclareSymbolFont{MnLargeSymbols}{OMX}{MnSymbolE}{m}{n}
 \SetSymbolFont{MnLargeSymbols}{bold}{OMX}{MnSymbolE}{b}{n}
 \DeclareFontShape{OMX}{MnSymbolE}{m}{n}{
	 <-6>  MnSymbolE5
	<6-7>  MnSymbolE6
	<7-8>  MnSymbolE7
	<8-9>  MnSymbolE8
	<9-10> MnSymbolE9
   <10-12> MnSymbolE10
   <12->   MnSymbolE12
 }{}
 \DeclareFontShape{OMX}{MnSymbolE}{b}{n}{
	 <-6>  MnSymbolE-Bold5
	<6-7>  MnSymbolE-Bold6
	<7-8>  MnSymbolE-Bold7
	<8-9>  MnSymbolE-Bold8
	<9-10> MnSymbolE-Bold9
   <10-12> MnSymbolE-Bold10
   <12->   MnSymbolE-Bold12
 }{}
 
 \let\llangle\@undefined
 \let\rrangle\@undefined
 \DeclareMathDelimiter{\llangle}{\mathopen}%
					  {MnLargeSymbols}{'164}{MnLargeSymbols}{'164}
 \DeclareMathDelimiter{\rrangle}{\mathclose}%
					  {MnLargeSymbols}{'171}{MnLargeSymbols}{'171}
\makeatother
\makeatletter
\newcommand{\dotr}[1]{%
  \mathpalette\@dotr{#1}%
}
\newcommand*{\@dotr}[2]{%
  \sbox0{$\m@th#1#2$}%
  \usebox{0}%
  \raisebox{\dimexpr\ht0-\height}{$\m@th#1\@smallbullet#1\bullet$}%
  \kern\scriptspace
}
\newcommand*{\@smallbullet}[2]{%
  \scalebox{.25}{$\m@th#1#2$}%
}
\makeatother

\newcommand{\const}{\,{\rm const}\,}

\newcommand{\td}{\text{d}}

\newcommand{\ord}{\mathcal{O}}

\title{\textbf{Uniqueness of the extremal Schwarzschild de Sitter spacetime}}
\author{David Katona\footnote{d.katona@sms.ed.ac.uk} \orcidlink{0000-0002-5710-0665}\,\, and \,James Lucietti\footnote{j.lucietti@ed.ac.uk} \orcidlink{0000-0002-2371-9689}
\\ \\ 
\small \sl School of Mathematics and Maxwell Institute for Mathematical Sciences, 
\\ 
\small \sl University of Edinburgh, King's Buildings, Edinburgh, EH9 3FD, UK }

\date{}

\begin{document}
	\maketitle
	\begin{abstract}
	We prove that any analytic vacuum spacetime with a positive cosmological  constant in four and higher dimensions, that contains a static extremal Killing horizon with a maximally symmetric compact cross-section, must be locally isometric to either the extremal Schwarzschild de Sitter solution or its near-horizon geometry (the Nariai solution). In four-dimensions, this implies these solutions are the only analytic vacuum spacetimes that contain a static extremal horizon with compact cross-sections (up to identifications).  We also consider the analogous uniqueness problem for the four-dimensional extremal hyperbolic Schwarzschild anti-de Sitter solution and show that it reduces to a spectral problem for the laplacian on compact hyperbolic surfaces, if a cohomological obstruction to the uniqueness of infinitesimal transverse deformations of the horizon is absent. 
 \end{abstract}	


\section{Introduction} \label{sec_intro} 

The celebrated no-hair theorem establishes uniqueness of asymptotically flat, stationary,  electro-vacuum black holes, under certain assumptions, see e.g.~\cite{chrusciel_stationary_2012}. The original theorems were established for non-extremal black holes, although in recent years these have been extended  to the extremal case after an improved understanding of their near-horizon geometry~\cite{amsel_uniqueness_2010, figueras_uniqueness_2009, kleinwachter_analytical_2008, chrusciel_uniqueness_2010}.  In fact, for static extremal black holes  or supersymmetric black holes, uniqueness fails if one allows for multiple black holes, with the general solution given by the Majumdar-Papapetrou solution in both cases~\cite{chrusciel_nonexistence_2005, chrusciel_israel-wilson-perjes_2006}. In higher dimensional general relativity, black hole uniqueness no longer holds even in vacuum gravity~\cite{emparan_rotating_2002, emparan_black_2008}. Nevertheless, a number of higher-dimensional black hole classification theorems have been derived under various  symmetry assumptions~\cite{hollands_black_2012}, with the most complete results known for static black holes~\cite{gibbons_uniqueness_2002, gibbons_dilatonic_2002, kunduri_static_2017, lucietti_higher_2020} and supersymmetric black holes~\cite{breunholder_moduli_2017, katona_supersymmetric_2022}.   

If a cosmological constant is present, direct analogues of the black hole uniqueness theorems are not known, except in a few limited cases. For a positive cosmological constant,  uniqueness theorems for the non-extremal Schwarzschild de Sitter (dS) black holes have been established within the class of static spacetimes, under various assumptions on the level sets of the lapse function~\cite{borghini_uniqueness_2023}.
Remarkably, recent numerical evidence  has been presented for the existence of static binary black holes in de Sitter, which evade these assumptions~\cite{dias_static_2023}. Therefore, even the classification of static black hole spacetimes in de Sitter is  not fully understood.   On the other hand, for negative cosmological constant, the only known black hole uniqueness theorem is for the nonpositive mass hyperbolic Schwarzschild anti-de Sitter (AdS) black holes~\cite{chrusciel_towards_2001,lee_penrose_2015}.  Uniqueness  of the spherical Schwarzschild-AdS black holes remains a notable open problem (see e.g.~\cite{chrusciel_nonsingular_2005,chrusciel_nondegeneracy_2017} and references therein for some results in this direction).

In this note we consider the classification of extremal black holes with a cosmological constant.
An important feature of extremal horizons is that they admit a well-defined near-horizon geometry that itself is a 
solution of the Einstein equations. Therefore, near-horizon geometries may be classified independently of any parent black hole spacetime. Indeed, many such near-horizon classifications have been established~\cite{kunduri_classification_2013}, even for solutions with a cosmological constant where the traditional methods for proving black hole uniqueness fail. However, not all near-horizon geometries are realised as near-horizon limits of 
black hole solutions, and even if they are, the corresponding solutions might not be unique. A natural question thus
arises: can one determine all spacetimes that contain an extremal horizon with a given near-horizon geometry? The systematic study of this question was initiated in~\cite{li_transverse_2016,li_electrovacuum_2019}, where the concept of  transverse 
deformations of an extremal horizon was introduced. This involves expanding the Einstein equation in a parameter that controls deformations of the metric that are transverse to the extremal horizon (i.e. in the direction away from the horizon). At first order in this expansion the deformations are governed by the linearised Einstein equations in the background near-horizon geometry, and have been largely determined under various symmetry assumptions in vacuum and electro-vacuum gravity in four-dimensions including a cosmological constant~\cite{li_transverse_2016,li_electrovacuum_2019,kolanowski_towards_2021} (see also~\cite{dunajski_einstein_2016, fontanella_moduli_2016}).
 
The purpose of this note is to show that for certain simple near-horizon geometries one can in fact determine all the higher order deformations. Therefore, for analytic spacetimes this allows one to deduce all possible exact transverse deformations. This offers a new method to establish uniqueness theorems for extremal black holes with prescribed near-horizon geometries.   
In particular, our main result is the following.

\begin{theorem}\label{thm_dS}
	Let $(M, g)$ be an analytic spacetime that obeys the $d\ge 4$-dimensional vacuum Einstein equation with cosmological constant
	$\Lambda>0$ and contains a static degenerate Killing horizon with a maximally symmetric compact cross-section. Then $(M, g)$ is locally isometric either to 
	the extremal Schwarzschild de Sitter solution or its near-horizon geometry (the Nariai solution).
\end{theorem}

This is the first uniqueness theorem for extremal vacuum black holes with a cosmological constant in four or higher dimensions.  The proof is elementary and runs as follows.  Previously, it has been shown that the first order transverse deformations for the near-horizon geometry dS$_2\times S^{d-2}$, also known as the Nariai solution,  are unique and if nonvanishing correspond to the first order deformations of the extremal Schwarzschild-dS solution~\cite{li_electrovacuum_2019}.  We show that this result persists at second order and via an inductive argument to all orders.  The key point is that at any order the Einstein equations are sourced by the lower order deformations and reduce to an eigenvalue equation for the laplacian on $S^{d-2}$ acting on a (traceless) part of the metric perturbation. For $\Lambda>0$ these eigenvalues are strictly negative and hence the Einstein equations only admit the trivial solution.

Our proof was inspired by a similar analysis of four-dimensional vacuum spacetimes with extremal toroidal horizons, which established that the only solution is a plane wave spacetime ~\cite{moncrief_symmetries_1983}. Their result can be interpreted as an explicit proof that  there are no four-dimensional extremal toroidal vacuum black holes.\footnote{Although their analysis assumes the generators of the horizon are periodic, this assumption is only used for proving existence of a Killing vector field tangent to the generators. Thus, their argument also applies to extremal Killing horizons with a toroidal cross-section, without any  assumptions on the horizon generators.}  
The idea of determining a spacetime from its near-horizon geometry alone has also been successfully applied to three-dimensional vacuum solutions with a cosmological constant~\cite{li_three_2014}, and five-dimensional supersymmetric black holes in AdS~\cite{lucietti_uniqueness_2021, lucietti_uniqueness_2022}, although in these cases there is no need to expand the Einstein equations order by order.

For $d=4$ spacetime dimensions, static near-horizon geometries with compact cross-sections have been shown to be unique~\cite{chrusciel_nonexistence_2005}.  Therefore, combining this  with Theorem \ref{thm_dS},  implies that any analytic spacetime containing a static extremal horizon with compact cross-sections must be locally isometric to either the extremal Schwarzschild de Sitter solution or the Nariai solution.   In particular, this solves the classification problem for static extremal vacuum black holes in de Sitter, assuming analyticity, although our result in fact only assumes staticity of the Killing field at the horizon and not globally.\footnote{By static we mean that the Killing field that is null on the horizon is everywhere hypersurface orthogonal, but not necessarily timelike anywhere. Indeed, for the extremal Schwarzschild-dS solution, this Killing field is strictly spacelike away from the horizon.}  

We emphasise that the above result  does not make any global assumptions on the spacetime such as asymptotics or the number of black holes. It therefore rules out the possibility of extremal multi-black holes in de Sitter at least for analytic spacetimes.  This is particularly interesting in view of the above mentioned fact that non-extremal static binary black holes have been recently constructed numerically~\cite{dias_static_2023}.  This should also be contrasted with the static extremal Majumdar-Papapetrou multi-black holes in Einstein-Maxwell theory (no cosmological constant), which are analytic in four-dimensions but not in higher dimensions~\cite{welch_smoothness_1995,candlish_smoothness_2007,lucietti_higher_2020}.  Furthermore, it has been argued that extremal black holes with a cosmological constant are generically not  smooth at the horizon~\cite{horowitz_almost_2022}.  It is therefore possible that extremal multi-black holes in de Sitter may exist under weaker differentiability assumptions.\footnote{These would be different to the dynamical multi-black holes in de Sitter~\cite{kastor_cosmological_1993}.}

The above method can also be applied to $\Lambda<0$ spacetimes that contain static extremal horizons with hyperbolic compact cross-sections\footnote{The spherical Schwarzschild-AdS black hole does not have an extremal limit.}. However, in contrast to the $\Lambda>0$ case, we find that if $\Lambda<0$ the transverse deformations are not unique even at first order; interestingly, the additional solutions correspond to harmonic 1-forms on the cross-section of the horizon. Therefore, the (first order) transverse deformations are determined by the first cohomology of the corresponding  compact hyperbolic surface of genus $g \ge 2$, which is $2g$-dimensional. Furthermore, even if we restrict to the trivial first-order deformation (these correspond to the extremal hyperbolic Schwarzschild-AdS solutions), we are unable to prove that the higher order deformations are unique as in the $\Lambda>0$ case. In particular, the Einstein equations also reduce to an eigenvalue equation for the laplacian, except for $\Lambda<0$ the eigenvalues are now positive. We find that a uniqueness theorem analogous to Theorem \ref{thm_dS}, for the extremal Schwarzschild-AdS solutions with a given compact hyperbolic cross-section, can be proven if and only if the first order deformations are trivial {\it and} the spectrum of the scalar laplacian on the corresponding compact hyperbolic surface does not contain any eigenvalues of the form $\lambda_n=n^2+n-2$ for integer $n\geq 2$ (if the  curvature is unit normalised). However, determining the spectrum of the laplacian on compact hyperbolic surfaces is an open problem, so we are unable to prove any definite uniqueness result for $\Lambda<0$. For certain special points in the moduli space of compact hyperbolic surfaces eigenvalues of this form can be realised~\cite{bonifacio_bootstrap_2021}, however,  presumably this is not the case generically and therefore we still expect uniqueness  to hold in this sense.  We discuss this further at the end of Section \ref{sec_deform}.

The organisation of this article is as follows. In Section \ref{sec_solutions} we review the Schwarzschild-(A)dS solutions.  In Section \ref{sec_deform} we introduce the notion of transverse deformations of extremal horizons and determine them to all orders for the extremal Schwarzschild-(A)dS horizons. In Appendix \ref{app_Ricci} we give the Ricci tensor in Gaussian null coordinates. In Appendix \ref{app_2tensor} we prove a decomposition theorem for traceless symmetric  $2$-tensors on compact hyperbolic surfaces.


\section{Extremal Schwarzschild (anti)-de Sitter  solutions} \label{sec_solutions} 

In this section we review the $d\geq 4$-dimensional Schwarzschild solutions with a cosmological constant $\Lambda$ and examine their extremal limits. Our conventions are such that the solutions satisfy the Einstein equation $R_{\mu\nu}= \Lambda g_{\mu\nu}$.

For $\Lambda>0$, the $d$-dimensional Schwarzschild-dS solutions are given by \cite{cardoso_nariai_2004}
\begin{equation}
	g = -\left(1 -\frac{m}{r^{d-3}}-\frac{\Lambda r^2}{d-1}\right)\td t^2 + \left(1 -\frac{m}{r^{d-3}}
	-\frac{\Lambda r^2}{d-1}\right)^{-1}\td r^2 + r^2 \td \Omega_{d-2}^2\;, \label{eq_SdSmetric}
\end{equation}
where $\td \Omega_{d-2}^2$ is the unit metric on the $(d-2)$-sphere $S^{d-2}$, and $0\le m\le m_{\max}$ is a mass parameter with
\begin{align}
	m_{\max} := \frac{2 r_0^{d-3}}{d-1}\;, \qquad  r_0 := \left(\frac{d-3}{\Lambda}\right)^{1/2}\;.
	\label{eq_mmax}
\end{align} 
For $m =0$ this gives de Sitter space, for $0<m<m_{\max}$ the Schwarzschild-dS black hole, which contains a black hole horizon and a cosmological horizon. For $m=m_{\max}$ one obtains the extremal Schwarzschild-dS 
solution, for which these two horizons coincide such that there is a degenerate horizon at $r= r_0$.
The Penrose diagram\footnote{One can also draw a time-reversed Penrose diagram, with the singularity at the bottom, and 
$\cal{J}^+$ at the top.} of its maximal analytic extension can be seen in Fig.~\ref{fig_eSdS_penrose}~\cite{lake_effects_1977,podolsky_structure_1999}. Note that $P$ are asymptotic points which can be 
reached by causal geodesics with $t=\const$ and for such observers $r= r_0$ is an event horizon. 
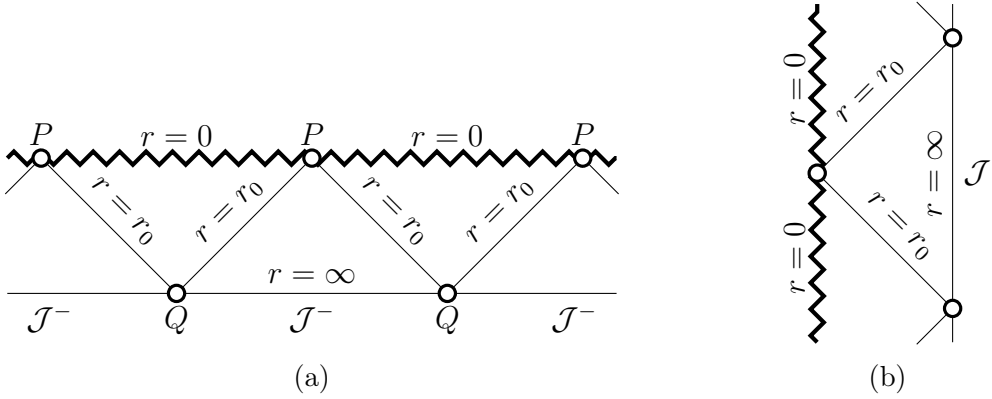
\begin{figure}
	\centering
	\begin{subfigure}[b]{0.6\textwidth}
		\centering
\begin{tikzpicture}[scale=0.90]

	\coordinate (BHa) at (-4, 2);
	\coordinate (BHb) at (0, 2);
	\coordinate (BHc) at (4, 2);
	\coordinate (LeftTop) at (-4.5, 2);
	\coordinate (RightTop) at (4.5, 2);

	\path (BHa) +(-45:2.8284)  coordinate  (BHaBottom);
	\path (BHb) +(-45:2.8284)  coordinate  (BHbBottom);
	\path (BHaBottom) +(180:2.5)  coordinate  (LeftBottom);
	\path (BHbBottom) +(0:2.5)  coordinate  (RightBottom);
	\path (BHa) + (-135:0.75) coordinate (LeftHorizonEnd);
	\path (BHc) + (-45:0.75) coordinate (RightHorizonEnd);

	\draw[decorate,decoration=zigzag, ultra thick] (LeftTop) --(BHa) -- node[midway, above] {$r=0$} (BHb) -- 
	node[midway, above] {$r=0$} (BHc)-- (RightTop);

	\draw (LeftHorizonEnd) -- (BHa) 
		-- node[midway, above, sloped] {$r= r_0$} (BHaBottom) 
		-- node[midway, above, sloped] {$r= r_0$} (BHb) 
		-- node[midway, above, sloped] {$r= r_0$} (BHbBottom) 
		-- node[midway, above, sloped] {$r= r_0$} (BHc)
		-- (RightHorizonEnd);

	\draw (LeftBottom) -- 
		node[near start, below] {$\cal{J}^-$}
		(BHaBottom) -- 
		node[midway, below] {$\cal{J}^-$} 
		node[midway, above] {$r = \infty$}
		(BHbBottom) -- 
		node[near end, below]{$\cal{J}^-$} 
		(RightBottom);

		\filldraw[color=black, fill=white, very thick](BHa) circle (0.12);
		\filldraw[color=black, fill=white, very thick](BHb) circle (0.12);
		\filldraw[color=black, fill=white, very thick](BHc) circle (0.12);
		\filldraw[color=black, fill=white, very thick](BHaBottom) circle (0.12);
		\filldraw[color=black, fill=white, very thick](BHbBottom) circle (0.12);

		\path (BHa) +(90:0.35) node (Ia) {$P$};
		\path (BHb) +(90:0.35) node (Ia) {$P$};
		\path (BHc) +(90:0.35) node (Ia) {$P$};
		\path (BHaBottom) +(-90:0.4) node (Ia) {$Q$};
		\path (BHbBottom) +(-90:0.4) node (Ia) {$Q$};

\end{tikzpicture}
\caption{}
\label{fig_eSdS_penrose}
\end{subfigure}
\begin{subfigure}[b]{0.3\textwidth}
	\centering
	\begin{tikzpicture}[scale=0.9]
		\coordinate (LeftBottom) at (-2, -2.5);
		\coordinate (BHb) at (-2, 0);
		\coordinate (LeftTop) at (-2, 2.5);

		\path (BHb) +(-45:2.8284)  coordinate  (BHaRight);
		\path (BHb) +(45:2.8284)  coordinate  (BHbRight);
		\path (BHaRight) +(-90:0.5)  coordinate  (RightBottom);
		\path (BHbRight) +(90:0.5)  coordinate  (RightTop);
		\path (BHaRight) + (-135:0.75) coordinate (BottomHorizonEnd);
		\path (BHbRight) + (135:0.75) coordinate (TopHorizonEnd);
		
		\draw[decorate,decoration=zigzag, ultra thick] (LeftBottom) --  node[midway, above, sloped]{$r=0$} (BHb) 
		--  node[midway, above, sloped]{$r=0$} (LeftTop);

		\draw (BottomHorizonEnd) -- (BHaRight) 
			-- node[midway, above, sloped] {$r= r_0$} (BHb) 
			-- node[midway, above, sloped] {$r= r_0$} (BHbRight) 
			-- (TopHorizonEnd);

		\draw (RightBottom) --  
			node[midway, right] {$\cal{J}$} 
			node[midway, above, sloped] {$r = \infty$}
			(RightTop);

		\filldraw[color=black, fill=white, very thick](BHb) circle (0.12);
		\filldraw[color=black, fill=white, very thick](BHbRight) circle (0.12);
		\filldraw[color=black, fill=white, very thick](BHaRight) circle (0.12);

	\end{tikzpicture}
	\caption{}
	\label{fig_eSAdS_penrose}
\end{subfigure}
\caption{Penrose diagrams for $(a)$ extremal Schwarzschild-dS~\cite{lake_effects_1977,podolsky_structure_1999} and $(b)$ extremal hyperbolic Schwarzschild-AdS solutions~\cite{mann_topological_1997}.}
\end{figure}

One can also take the $m\to m_{\max}$ limit and simultaneously `blow up' the outer region between the cosmological and 
black hole horizons~\cite{ginsparg_semiclassical_1983} to obtain the Nariai 
solution~\cite{nariai_new_1999,nariai_static_1999}. This is simply the vacuum solution dS$_2\times S^{d-2}$ with constant 
radii, which also arises as the near-horizon geometry of the extremal Schwarzschild-dS solution.

For $\Lambda<0$ the hyperbolic Schwarzschild-AdS solution is given by 
\begin{equation}
	g = -\left(-1 -\frac{m}{r^{d-3}}-\frac{\Lambda r^2}{d-1}\right)\td t^2 + \left(-1 -\frac{m}{r^{d-3}}
	-\frac{\Lambda r^2}{d-1}\right)^{-1}\td r^2 + r^2 \td \Sigma_{d-2}^2\;, \label{eq_SAdSmetric}
\end{equation}
where $\td \Sigma_{d-2}^2$ is the metric of $(d-2)$-dimensional unit hyperbolic space $H^{d-2}$ and $m_{\min}\le m$ with 
\begin{align}
	m_{\min}:=-\frac{2 r_0^{d-3}}{d-1}\;,\qquad  r_0 := \left(\frac{d-3}{|\Lambda|}\right)^{1/2}\;.
	\label{eq_mmin}
\end{align} 
(Note that $m_{\min}<0$.) To obtain horizons with compact cross-sections one takes a discrete quotient of the $d-2$ dimensional 
hyperbolic space. For $d=4$, it is known that one can obtain compact hyperbolic surfaces with any genus $g\ge 2$, hence 
these black holes are also known as topological black holes. In the extremal case $m=m_{\min}$ one obtains a solution 
with a degenerate horizon at $r= r_0$ and its Penrose diagram is depicted in Fig.~\ref{fig_eSAdS_penrose}~\cite{mann_topological_1997}. Its near-horizon geometry is AdS$_2\times H^{d-2}$ with constant radii, 
sometimes called the anti-Nariai solution~\cite{caldarelli_extremal_2000}.

The extremal Schwarzschild-dS and hyperbolic Schwarzschild-AdS solutions can be written in a unified way as 
\begin{equation}
	g= - s f(r)\td t^2 + s f(r)^{-1}\td r^2 + r^2 \gamma_s \;, \label{eq_metric_general}
\end{equation}
with $s=\pm1$, where $\gamma_1= \td \Omega_{d-2}^2$ and $\gamma_{-1}= \td \Sigma_{d-2}^2$, so $\text{Ric}(\gamma_s)= s (d-3)\gamma_s$, and 
\begin{equation}
	f(r):=1  - \frac{2}{d-1}\left(\frac{ r_0}{r}\right)^{d-3}-\frac{d-3}{d-1}\left(\frac{r}{ r_0}\right)^2\; .
	\label{eq_fr}
\end{equation}
It is easy to check that $f(r)$ has a double zero at $r= r_0$ and can be written as
\begin{equation}
	f(r) = -\frac{(r- r_0)^2}{(d-1) r_0^2}\left[(d-3)+
	2\sum_{k=1}^{d-3}k\left(\frac{r}{ r_0}\right)^{k-d+2}\right]\; ,  \label{eq_fdouble0}
\end{equation}
which also shows that it is strictly negative away from the horizon. Thus, in the $s=1$ case (extremal Schwarzschild-dS) the Killing vector field $\partial_t$ is spacelike outside the horizon and $r$ acquires the interpretation of a time coordinate.

\section{Transverse deformations of static extremal  horizons} \label{sec_deform} 

\subsection{Einstein equations near an extremal horizon}

Let $(M,g)$ be a $d$-dimensional spacetime satisfying the Einstein equation
\begin{equation}
R_{\mu\nu} = \Lambda g_{\mu\nu}\;, \label{eq_Einstein} 
\end{equation}
that contains a smooth degenerate (extremal) Killing horizon $\mathcal{H}$ of a Killing field $K$ with a compact spacelike cross-section $S$.  In a neighbourhood of  $\mathcal{H}$ we  introduce Gaussian null 
coordinates $(v, \rho, x^a)$ in terms of which the metric takes the form~\cite{moncrief_symmetries_1983, kunduri_classification_2013}, 
\begin{equation}
	g = \phi \td v^2 +2\td v \td \rho + 2\beta_a \td x^a \td v+ \mu_{ab}\td x^a \td x^b\;, 
	\label{eq_metric_Gaussian}
\end{equation}
where $K=\partial_v$ is the Killing field,  $L:= \partial_\rho$ is a transverse null geodesic field synchronised so the horizon is at $\rho=0$,  $(x^a)$ are an arbitrary chart on $S$, and  $\phi, \beta_a$  vanish at the horizon $\rho=0$. Furthermore, degeneracy of the horizon implies that $\partial_\rho \phi$ must also vanish at the horizon, so we can introduce smooth functions $F, h_a$, such that   
\begin{align}
	\phi=: \rho^2F\;, \qquad \beta_a =: \rho h_a\;. \label{eq_redef}
\end{align}
This coordinate system is unique up to a choice of cross-section $S$ and coordinates on $S$. The quantities $F$, $h=h_a \td x^a$, $\mu=\mu_{ab}\td x^a \td x^b$ can be identified with a function, 1-form and Riemannian metric on the codimension-2 surfaces $S_{v,\rho}$ of constant $(v, \rho)$ which include the cross-section $S$ at $\rho=0$.

The near-horizon geometry is defined as follows. Consider the 1-parameter family of metrics $g_{\epsilon}:= \phi_\epsilon^* g$ defined by the scaling diffeomorphism $\phi_\epsilon: (v, \rho, x^a)\mapsto (v/\epsilon, \epsilon\rho, x^a)$ with $\epsilon >0$.  The near-horizon geometry is the limit $\epsilon\to 0$ of $g_\epsilon$,  which gives
\begin{equation}
	g^{[0]} = \rho^2\mathring{F}(x)\td v^2 +2\td v \td \rho + 2\rho\mathring h_a(x) \td x^a \td v+ \mathring \mu_{ab}(x)\td x^a \td x^b \;,
\end{equation}
where the superscript $\circ$ denotes the value of the quantity at $\rho=0$, i.e. $\mathring F(x)= F(0, x)$ etc. The near-horizon data $(\mathring F, \mathring h, \mathring \mu)$ correspond to a function, 1-form and Riemannian metric on $S$.   The Einstein equation \eqref{eq_Einstein} for the near-horizon geometry (which itself must be a solution),  or the restriction of the Einstein equation for full spacetime metric \eqref{eq_metric_Gaussian} to $\rho=0$, are equivalent to the following equations for the near-horizon data on $S$,
\begin{equation}
\mathring{\mathcal{R}}_{ab}=\frac{1}{2} \mathring h_a \mathring h_b -\mathring \nabla_{(a} \mathring h_{b)} + \Lambda \mathring \mu_{ab}, \qquad \mathring F= \frac{1}{2} \mathring h^a \mathring h_a- \frac{1}{2}\mathring\nabla_a \mathring h^a+ \Lambda  \; ,
\end{equation}
where $\mathring \nabla_a, \mathring{\mathcal{R}}_{ab}$ are the metric connection and Ricci tensor of the metric $\mathring\mu$ on $S$.  The classification of solutions to these horizon equations has been extensively studied in the literature~\cite{kunduri_classification_2013}. 

The first order {\it  transverse deformation}  of an extremal horizon was introduced in~\cite{li_transverse_2016}. It is defined by $g^{[1]}:= \frac{\td}{\td \epsilon} g_\epsilon|_{\epsilon=0}$ and encodes the first transverse derivatives of the metric data at the horizon, that is, $F^{(1)}:=( \partial_\rho F)_{\rho=0}$, $h_a^{(1)} :=( \partial_\rho h_a )_{\rho=0}$ and $\mu_{ab}^{(1)} :=( \partial_\rho \mu_{ab})_{\rho=0}$.   The Einstein equation \eqref{eq_Einstein} implies $g^{[1]}$ satisfies a linearised Einstein equation in the background near-horizon geometry and it was shown that this is equivalent to  $\mu_{ab}^{(1)}$ satisfying a linear elliptic PDE (once gauge fixed) on $(S, \mathring \mu)$ with the rest of the first order data then determined algebraically.   

We define the higher order transverse deformations of an extremal horizon similarly by 
\begin{equation}
g^{[n]} := \left. \frac{\td^n  g_\epsilon}{\td \epsilon^n} \right|_{\epsilon=0}
\end{equation}
so $g^{[0]}$ agrees with the near-horizon geometry, $g^{[1]}$ with the first order deformation, and $n\geq 2$ gives the higher order deformations.  Explicitly, the $n$-th order deformation $g^{[n]}$ encodes the $n$-th transverse derivatives at the horizon of the metric data, that is, 
\begin{equation}
F^{(n)}:=( \partial^n_\rho F)_{\rho=0}, \qquad h_a^{(n)} :=( \partial^n_\rho h_a )_{\rho=0}, \qquad \mu_{ab}^{(n)} :=( \partial^n_\rho \mu_{ab})_{\rho=0}  \; .
\end{equation}
(Note that this is equivalent to $\phi^{(n+2)}, \beta^{(n+1)}_a, \mu^{(n)}_{ab}$).
In general for any function $X$, we will adopt the notation $X^{(n)}:= (\partial_\rho^n X)_{\rho=0}$ for $n \ge 0$. Therefore, for analytic spacetimes,  if all higher order transverse deformations  of an extremal horizon are known,  the exact solution is fully determined.  

The Einstein equation \eqref{eq_Einstein}  implies that the $n$-th order transverse deformation satisfies $\text{Ric}^{[n]}(g)= \Lambda g^{[n]} $ for all $n\geq 0$, where $\text{Ric}^{[n]}(g):=   \frac{\td^n}{\td \epsilon^n} \text{Ric}(g^\epsilon)|_{\epsilon=0}$. In fact, we will implement the Einstein equations more directly and simply evaluate the $\rho$-derivatives at $\rho=0$ of the components of the Einstein equation \eqref{eq_Einstein} in Gaussian null coordinates, that is, 
\begin{equation}
R_{\mu\nu}^{(n)}= \Lambda g_{\mu\nu}^{(n)}  \; ,  \label{eq_Einstein_nth}
\end{equation}
where the notation is as above so $R_{\mu\nu}^{(n)} := (\partial_\rho^nR_{\mu\nu})_{\rho=0}$ and similarly for $g_{\mu\nu}$.
For this computation we need the Ricci tensor for a general metric in Gaussian null coordinates \eqref{eq_metric_Gaussian}. This has been written down in a convenient form in~\cite{moncrief_symmetries_1983} (see also  the Appendix of \cite{hollands_stationary_2009}).
The components of the Ricci tensor of \eqref{eq_metric_Gaussian} relevant for our calculation are given in Appendix \ref{app_Ricci}.

\subsection{Horizon geometry and first order transverse deformations} 

For every dimension $d\geq 4$, we  consider the following  near-horizon data
\begin{equation}
\mathring F= \Lambda, \qquad \mathring h_a=0, \qquad   \;   \mathcal{\mathring{R}}_{abcd}= \frac{\Lambda}{d-3} (\mathring \mu_{ac}\mathring\mu_{bd}- \mathring\mu_{ad}\mathring\mu_{bc})  \; ,\label{eq_NHdata}
\end{equation}
where $\mathcal{\mathring{R}}_{abcd}$ is the Riemann tensor of $(S, \mathring \mu)$, so in particular $\mathcal{\mathring{R}}_{ab}= \Lambda \mathring{\mu}_{ab}$, which corresponds to the most general static near-horizon geometry with maximally symmetric compact cross-sections $S$~\cite{chrusciel_nonexistence_2005, bahuaud_static_2022, wylie_rigidity_2023}. In fact, for $d=4$ and any $\Lambda$, it has been shown that this is the most general static near-horizon geometry with compact $S$~\cite{chrusciel_nonexistence_2005}. Furthermore,  \eqref{eq_NHdata} corresponds to the near-horizon data of the $ d\geq 4$ extremal Schwarzschild-(A)dS solution (\ref{eq_metric_general}). To see this, switch to (ingoing) 
Eddington-Finkelstein coordinates defined by $\td v = \td t + s f(r)^{-1}\td r$, $\rho=r- r_0$, so (\ref{eq_metric_general})  becomes
\begin{equation}
	g = -s f(r_0+\rho)\td v^2 + 2\td v \td \rho + ( r_0+\rho)^2\gamma_s \label{eq_metric_EF} \; ,
\end{equation}
and using (\ref{eq_fdouble0}) one can immediately read off that the near-horizon data  is \eqref{eq_NHdata} with $\mathring\mu = r_0^2\gamma_s$.

For the near-horizon geometry \eqref{eq_NHdata} with $\Lambda>0$, the first order transverse deformations have been determined~\cite{li_electrovacuum_2019}. For completeness, we will now give an alternate derivation of this result.

\begin{prop}\label{prop_1st}
Consider a spacetime that satisfies the Einstein equation \eqref{eq_Einstein} and contains an extremal horizon with a compact cross-section $S$ with near-horizon data \eqref{eq_NHdata}. If $d \geq 4$ and $\Lambda>0$, the first order transverse deformations are given by
\begin{align}
	\mu^{(1)}_{ab}=C\mathring\mu_{ab}\;, \qquad \beta_a^{(2)}=0\;, \qquad \phi^{(3)}=-  C(d-2)\Lambda \; ,    \label{eq_1st}
\end{align}
where $C$ is an arbitrary constant.
\end{prop}

\begin{proof} 
As discussed in~\cite{li_transverse_2016}, there is a gauge freedom that leaves the near-horizon data invariant, but 
changes the first (and higher) order deformations. This freedom corresponds to a change in the spatial cross-section 
$S$, analogous to `supertranslations' in asymptotic symmetries.  For the horizon data \eqref{eq_NHdata}, the first transverse deformation 
of  $\mu_{ab}$ transforms under such a gauge transformation as 
\begin{equation}
	\mu'^{(1)}_{ab} =\mu^{(1)}_{ab} + \mathring\nabla_a\mathring\nabla_b f\;,  
\end{equation}
where $f$ is an arbitrary function on $S$ that generates the gauge transformation. We may fix this gauge by requiring that 
\begin{equation}
	\mathring{\mu}^{ab}\mu^{(1)}_{ab} = C(d-2) \; , \label{eq_Cdef}
\end{equation}
for some constant $C$, which is always possible by existence 
results for the Poisson-equation on compact manifolds. Note that this completely fixes the Gaussian null coordinates 
(up to choice of coordinates on $S$).

Now we implement the Einstein equations (\ref{eq_Einstein_nth}) for the first order deformation of the near-horizon data \eqref{eq_NHdata}, where recall that the Ricci tensor in Gaussian null coordinates is given in  Appendix \ref{app_Ricci}. We find that the Einstein equations ${R}^{(0)}_{\rho a}=0$ and $R^{(1)}_{v\rho}=0$  
yield
\begin{align}
	\phi^{(3)} &= -C(d-2)\Lambda - \mathring\nabla^a\beta_a^{(2)}\label{eq_F_1}  \; ,\\
	\beta_a^{(2)} &= -\mathring\nabla^b\mu^{(1)}_{ab}\,\label{eq_h_1}  \; ,
\end{align}
respectively,  hence ${\phi}^{(3)}$ and $\beta^{(2)}$ (or by \eqref{eq_redef} equivalently $F^{(1)}$ and $h^{(1)}$) are determined by $\mu^{(1)}_{ab}$. 
Next, we find that the Einstein equation $R^{(1)}_{ab}=\Lambda\mu^{(1)}_{ab}$ gives 
\begin{equation}
	\mathcal{R}^{(1)}_{ab} = \mathring\nabla_{(a} \mathring\nabla^c\mu^{(1)}_{b)c}\;, \label{eq_mu_1}
\end{equation}
where $\mathcal{R}_{ab}$ is the Ricci tensor of the codimension-2 surfaces $(S_{v,\rho}, \mu)$ and we have used (\ref{eq_NHdata}) and (\ref{eq_h_1}) to eliminate $\beta^{(2)}_a$.
The variation of the Ricci tensor is given by
\begin{equation}
	\mathcal{R}^{(1)}_{ab} = \mathring\Delta_L\mu^{(1)}_{ab} + \mathring\nabla_{(a}v^{(1)}_{b)}
	\;,\label{eq_ricci_pert}
\end{equation}
where $\mathring\Delta_L$ is the Lichnerowicz operator of $(S, \mathring \mu)$ which explicitly is
\begin{equation}
	\mathring\Delta_L\mu^{(1)}_{ab}:= -\frac{1}{2}\mathring\nabla^2\mu^{(1)}_{ab} - \mathcal{\mathring{R}}_a{}^c{}_b{}^d\mu^{(1)}_{cd} 
	+ \mathcal{\mathring{R}}_{(a}^c\mu^{(1)}_{b)c}\;,
	\label{eq_def_lichnerowicz}
\end{equation}
and 
\begin{equation}
	v^{(1)}_b:= \mathring\nabla^c\mu^{(1)}_{bc} - \frac{1}{2}\mathring\nabla_b\left(\mathring\mu^{cd}\mu^{(1)}_{cd}\right)\;. 
	\label{eq_def_v}
\end{equation}
Finally, using (\ref{eq_ricci_pert}), (\ref{eq_Cdef}) and \eqref{eq_NHdata}, equation (\ref{eq_mu_1}) reduces to 
\begin{equation}
	-\mathring \nabla^2  \hat\mu^{(1)}_{ab}  = -2\Lambda\frac{d-2}{d-3} \hat\mu^{(1)}_{ab}\;, \label{eq_ev_1}
\end{equation}
where $\hat\mu^{(1)}_{ab}:= \mu^{(1)}_{ab}-C\mu_{ab}$ is the traceless part of $\mu^{(1)}_{ab}$.

The laplacian $-\mathring\nabla^2$ is positive-definite for a compact manifold, hence for $\Lambda>0$ the only solution is $\hat{\mu}^{(1)}_{ab}=0$. Therefore, for the $d\ge4$ extremal Schwarzschild-dS horizon ($\Lambda>0$) the first order transverse deformations are  given by \eqref{eq_1st}, where the rest of the first order data is fixed by \eqref{eq_F_1}, \eqref{eq_h_1}, which establishes Proposition \ref{prop_1st}. 
\end{proof}

This result shows that the only nonvanishing first order transverse deformations to the extremal Schwarzschild-dS horizon are those corresponding to the full extremal Schwarzschild-dS solution.  

The above argument fails for $\Lambda<0$ because the eigenvalue in (\ref{eq_ev_1}) are in this case positive. In the four-dimensional case, the first order deformations are given as follows.
\begin{prop}\label{prop_AdS_1}
	Consider a spacetime that satisfies the Einstein equation \eqref{eq_Einstein} and contains an extremal horizon with a compact cross-section $S$ with near-horizon data \eqref{eq_NHdata}. If $d=4$ and $\Lambda<0$, the first order deformations are given by 
	\begin{align}
		\mu^{(1)}_{ab} = C\mathring\mu_{ab} + \sum_{i=1}^{2g}C_i\mathring\nabla_{(a}\xi^i_{b)}\;, \qquad \beta^{(2)}_a= -\Lambda \sum_{i=1}^{2g}C_i\xi^i_a\;, \qquad \phi^{(3)} = -C(d-2)\Lambda\;,
	\end{align}
	where $g\ge2$ is the genus of the hyperbolic surface $S$, $C$ and $C_i$ are constants, and $\{\xi^i\}_{i=1}^{2g}$ is a basis of harmonic 1-forms on $S$ with respect to the Hodge--de Rham laplacian.
\end{prop}

\begin{proof}
As in the $\Lambda>0$ case, the first order deformations are determined by solutions to (\ref{eq_ev_1}) through (\ref{eq_F_1}-\ref{eq_h_1}). In Appendix \ref{app_2tensor} we show that the space of traceless symmetric  $(0,2)$ tensor fields on $S$ is spanned by 
\begin{align}
	S^{\{\lambda+4\}}_{ab} &:= \mathring{\nabla}_a\mathring{\nabla}_b f^{\{\lambda\}} - \frac{1}{2}\mathring \mu_{ab}\mathring\nabla^2f^{\{\lambda\}}\;, \qquad
	P^{\{\lambda+4\}}_{ab} := \mathring{\nabla}_{(a}\mathring\epsilon_{b)c}\mathring\nabla^cf^{\{\lambda\}}\;, \nonumber\\
	V^{\{4\}, i}_{ab} &:= \mathring{\nabla}_{(a}\xi^i_{b)}\;, \qquad\qquad\qquad\qquad\qquad Y^{\{2\}}_{ab}\;,\label{eq_2tensors}
\end{align}
where $\mathring\epsilon$ is the volume form of $(S, \mathring\mu)$, $f^{\{\lambda\}}$ are eigenfunctions of $-\mathring\nabla^2$ with 
eigenvalues $\lambda |\Lambda| > 0$, $\{\xi^i\}_{i=1}^{2g}$ is a basis of harmonic 1-forms on $S$ with respect to the Hodge--de Rham laplacian $\Delta=-\mathring{\nabla}^2+\Lambda$, and $Y^{\{2\}}$ is a  divergence-free traceless symmetric  $(0,2)$ tensor.  Note that the space of harmonic 1-forms for a genus $g$ hyperbolic surface is $2g$-dimensional.   Furthermore,  $S^{\{\lambda+4\}}, P^{\{\lambda+4\}}, V^{\{4\}}, Y^{\{2\}}$ are eigentensors of $-\mathring\nabla^2$ with eigenvalues $(\lambda+4)|\Lambda|$, $(\lambda+4)|\Lambda|$, $4|\Lambda|$ and $2|\Lambda|$, respectively. Now, comparison with (\ref{eq_ev_1}) yields that $\hat\mu^{(1)}_{ab}$ must be in the span of $V_{ab}^{\{4\}}$.  The  form of $\beta^{(2)}$ and $\phi^{(3)}$ follows from (\ref{eq_F_1}-\ref{eq_h_1}).
\end{proof}

The first order deformations with $C_i=0$ in Proposition \ref{prop_AdS_1} correspond to those of the extremal hyperbolic Schwarzschild-AdS solution. Therefore, the above result shows that the first order transverse deformations of the extremal hyperbolic Schwarzschild-AdS horizon are not unique, but parameterised by harmonic 1-forms on $S$. Interestingly, the first cohomology of $S$ provides an obstruction to the uniqueness of these deformations. This non-uniqueness was not found in \cite{li_electrovacuum_2019} since that work did not consider the tensors $V^{\{4\}}$.

\subsection{Second and higher order transverse deformations}\label{ssec_second} 

We now consider higher order transverse deformations of extremal horizons with near-horizon data \eqref{eq_NHdata} and first order data \eqref{eq_1st}.   It is convenient to first consider the second order deformations. 
\begin{prop}\label{prop_2nd}
Consider a spacetime with an extremal horizon as in Proposition \ref{prop_1st} and $\Lambda>0$. The second order transverse deformations are given by
\begin{align}
\mu^{(2)}_{ab}=\frac{C^2}{2}\mathring \mu_{ab}	\;, \qquad \beta_a^{(3)}=0\;, \qquad \phi^{(4)}= \frac{1}{2} C^2\Lambda d (d-2) \; .  \label{eq_2nd}
\end{align}
\end{prop}

\begin{proof}
Consider  the Einstein equation \eqref{eq_Einstein_nth} for the second order  deformation of the near-horizon data \eqref{eq_NHdata}, assuming the first order data \eqref{eq_1st}, where recall the components of the Ricci tensor are given in Appendix \ref{app_Ricci}.  Firstly, we note that
${R}^{(0)}_{\rho\rho}=0$ reduces to 
\begin{equation}
	\mathring\mu^{ab}\mu^{(2)}_{ab} = \frac{1}{2}C^2(d-2)\; , \label{eq_trace_2}
\end{equation}
where we have used \eqref{eq_Cdef}, which shows that the trace of the second order  deformation of $\mu_{ab}$ is also a constant. Next, we find that the Einstein equations $R^{(2)}_{v\rho}=0$ and $R^{(1)}_{\rho a}=0$ reduce to
\begin{align}
	\phi^{(4)} &= \frac{1}{2}C^2\Lambda d(d-2)-\mathring\nabla^a\beta_a^{(3)}\;,\label{eq_F_2}\\
	\beta_a^{(3)}&=-\mathring\nabla^b\mu^{(2)}_{ab}\; ,\label{eq_h_2}
\end{align}
respectively, where we have used  (\ref{eq_trace_2}) together with the lower order data. Thus analogously to the first order deformations, we deduce that  $\phi^{(4)}$ and 
$\beta^{(3)}$ (that is $F^{(2)}$ and $h^{(2)}$) are determined in terms of $\mu^{(2)}$.
The Einstein equation $R^{(2)}_{ab}=\Lambda\mu^{(2)}_{ab}$ yields
\begin{equation}
	 \mathcal{R}^{(2)}_{ab}+2\Lambda \mu^{(2)}_{ab}-\Lambda C^2\mathring\mu_{ab}+ \mathring\nabla_{(a}\beta_{b)}^{(3)} 
	=0\;.\label{eq_mu_2}
\end{equation}
We may evaluate the second variation of the Ricci tensor $\mathcal{R}^{(2)}_{ab}$ as follows. 

First note that we have so far shown that
\begin{equation}
	\mu_{ab} = \mathring\mu_{ab} + C\rho\mathring\mu_{ab} + \frac{1}{2}\rho^2\mu^{(2)}_{ab} + \ord(\rho^3),
\end{equation}
Now  define a rescaled metric by $\mu=(1+C\rho)\tilde\mu$ so that
\begin{equation}
\tilde \mu _{ab}= \mathring\mu_{ab} + \varepsilon\mu^{(2)}_{ab}+ \ord\left(\varepsilon^{3/2}\right)  \; ,
\end{equation}
where  we have introduced a new
expansion parameter  $\varepsilon:=\frac{1}{2}\frac{\rho^2}{1+C\rho}$.  Thus we can expand the Ricci tensor of $\mu_{ab}$ as 
\begin{equation}
	\mathcal{R}_{ab} = \tilde{\mathcal{R}}_{ab} = \mathring{\mathcal{R}}_{ab} 
	+ \varepsilon \left( \frac{\td}{\td\varepsilon}\tilde{\mathcal{R}}_{ab}\right)_{\varepsilon=0} + \ord\left(\varepsilon^{3/2}\right)\;,
	\label{eq_ricci_2}
\end{equation}
where $\tilde{\mathcal{R}}_{ab}$ is the Ricci tensor of $\tilde\mu_{ab}$.
Thus, recalling that $\mathcal{R}^{(2)}_{\mu\nu}= (\partial^2_\rho \mathcal{R}_{\mu\nu})_{\rho=0}$ and $\varepsilon= \tfrac{1}{2}\rho^2+O(\rho^3)$, we deduce that the second variation of the Ricci tensor is given by the usual first variation of the Ricci tensor applied to the second variation of the metric, that is, 
\begin{equation}
	\mathcal{R}_{ab}^{(2)} = \frac{\td}{\td\varepsilon}\tilde{\mathcal{R}}_{ab}\bigg|_{\varepsilon=0} = 
	\mathring\Delta_L\mu^{(2)}_{ab} + \mathring\nabla_{(a}v^{(2)}_{b)}\;, \label{eq_R_pert}
\end{equation}
where $v^{(2)}_a=\mathring\nabla^b\mu^{(2)}_{ba}- \frac{1}{2} \mathring\nabla_a (\mathring\mu^{cd}\mu^{(2)}_{cd})$.

We now have all the required ingredients. Using (\ref{eq_R_pert}), (\ref{eq_h_2}), (\ref{eq_trace_2}) and \eqref{eq_NHdata}, we find that (\ref{eq_mu_2}) simplifies to 
\begin{equation}
	- \mathring\nabla^2 \hat\mu^{(2)}_{ab}= 
	-2\Lambda\frac{3d-8}{d-3} \hat\mu^{(2)}_{ab}\;, \label{eq_ev_2}
\end{equation}
where $\hat\mu^{(2)}_{ab}:= \mu^{(2)}_{ab}-\frac{C^2}{2}\mathring \mu_{ab}$ is the traceless part of $\mu^{(2)}_{ab}$ (recall \eqref{eq_trace_2}). Therefore, for $\Lambda>0$ we again deduce that the  solution is unique and in this case given by $\hat{\mu}^{(2)}_{ab}=0$.  The rest of the second order data is fixed by \eqref{eq_F_2}, \eqref{eq_h_2} which gives \eqref{eq_2nd}.
\end{proof}

This shows that for the $d\geq 4$ extremal Schwarzschild-dS horizon  the nonzero second order transverse deformations are also uniquely given by the full extremal Schwarzschild-dS solution. 

We will now turn to higher order deformations. Our main result is the following.
\begin{prop} \label{prop_nth}
Consider a spacetime with an extremal horizon as in Proposition \ref{prop_1st} and $\Lambda>0$. The $n$-th order transverse deformations for $n\geq 3$ are given by
\begin{align}
\mu^{(n)}_{ab}=0	\;, \qquad \beta_a^{(n+1)}=0\;, \qquad \phi^{(n+2)}= (-1)^n\frac{2}{d-1}\frac{(d-2+n)!}{(d-3)!}\left(\frac{C}{2}\right)^{n}\Lambda \; .  \label{eq_nth}
\end{align}
\end{prop}

\begin{proof}
We prove this by induction.  Thus let $n\geq 3$ and for the induction hypothesis assume that to $(n-1)$-th order 
\begin{align}
	\mu_{ab}^{(k)}	&= \left\{
		\begin{array}{ll}
		\mathring\mu_{ab} \;  & \text{if } k =0 \\
			C\mathring\mu_{ab}\;, &  \text{if } k = 1 \\
			\frac{C^2}{2}\mathring\mu_{ab}\;, & \text{if } k = 2\\
			0\;, & \text{if } 3\leq k\leq n-1 ,
		\end{array}\right.\label{eq_ind_mu}
		\\
  \beta^{(k)} &=0\;,  \qquad \text{if } 0\leq k\leq n \label{eq_ind_h}\\
	\phi^{(k)} &= \left\{
		\begin{array}{ll}
			0\;, & \text{if } 0\leq k \leq 1 \\
			\left[\frac{2\delta_{2k}}{d-1}+(-1)^k\frac{2}{d-1}\frac{(d-4+k)!}{(d-3)!}
			\left(\frac{C}{2}\right)^{k-2}\right]\Lambda\;, & \text{if } 2 \leq k \leq  n+1  \; .
		\end{array}
	\right.\label{eq_ind_F}
\end{align}
The base case $n=3$ is established in Propositions \ref{prop_1st} and \ref{prop_2nd}.

The method is identical to that for the second order calculation.  We implement the Einstein equation \eqref{eq_Einstein_nth} for the $n$-th order data assuming the above induction hypothesis, again using the components of the Ricci tensor in Appendix \ref{app_Ricci}. Firstly, we find $R^{(n-2)}_{\rho\rho}=0$ determines the trace of $n$-th order deformation of $\mu_{ab}$ to be  
\begin{equation}
	\mathring \mu^{ab}\mu^{(n)}_{ab}=0\;.  \label{eq_ntrace}
\end{equation}
Then $R^{(n)}_{v\rho}=0$ and $R^{(n-1)}_{\rho a}=0$ determine $\phi^{(n+2)}$ and $\beta^{(n+1)}$ ($F^{(n)}$ and 
$h^{(n)}$) uniquely in terms of $\mu^{(n)}$  to be 
\begin{align}
	\phi^{(n+2)} &= -\mathring\nabla^a\beta_a^{(n+1)} 
			+\frac{2\Lambda}{d-1}\left(-\frac{C}{2}\right)^n\frac{(d-2+n)!}{(d-3)!} \;,\label{eq_F_n}\\
	\beta_a^{(n+1)}&=-\mathring\nabla^b\mu^{(n)}_{ab}\; ,\label{eq_h_n}
\end{align}
where we have used (\ref{eq_ind_mu}-\ref{eq_ind_F}) for the lower order terms.
Next, $R^{(n)}_{ab}=\Lambda\mu^{(n)}_{ab}$ yields 
\begin{equation}
	\mathring\nabla_{(a}\beta_{b)}^{(n+1)} + \mathcal{R}^{(n)}_{ab} -\Lambda\mu_{ab}^{(n)}
	+ \frac{1}{2}(\phi\dot\mu_{ab})^{(n+1)}
	=-\frac{1}{4}\left[\phi \mu^{cd}(\dot \mu_{cd}\dot \mu_{ab}-2\dot\mu_{ac}\dot\mu_{bd})\right]^{(n)}\;,
	\label{eq_mu_n}
\end{equation}
where $\cdot$ denotes the $\rho$-derivative (without evaluating at $\rho=0$). To evaluate the right-hand 
side of (\ref{eq_mu_n}) it is useful to note that the induction hypothesis (\ref{eq_ind_mu}) is equivalent to
\begin{equation}
	\mu_{ab} =\left(1+\frac{C\rho}{2}\right)^2 \mathring\mu_{ab} + \ord(\rho^n)\;.\label{eq_mu_explicit}
\end{equation}
Indeed, upon substituting (\ref{eq_mu_explicit}) and its inverse into (\ref{eq_mu_n}), the $\rho$-dependent prefactors cancel thus simplifying the 
calculation. Using (\ref{eq_ind_mu}), we 
can evaluate
\begin{equation}
	\frac{1}{2}(\phi\dot\mu_{ab})^{(n+1)} = 
	\left[\frac{C}{2}\phi^{(n+1)}+ \frac{n+1}{4}C^2\phi^{(n)}\right]\mathring\mu_{ab} 
	+ \frac{n(n+1)}{2}\Lambda\mu^{(n)}_{ab}\;. \label{eq_Fmu_derivative}
\end{equation}
Substituting (\ref{eq_Fmu_derivative}) into (\ref{eq_mu_n}), the lower order terms (right-hand side of (\ref{eq_mu_n}) 
and first term of (\ref{eq_Fmu_derivative})) cancel by the induction hypothesis (\ref{eq_ind_mu}-\ref{eq_ind_F}).  
Thus, using (\ref{eq_h_n}), (\ref{eq_mu_n}) 
reduces to
\begin{equation} \label{eq_nRab}
	 \mathcal{R}^{(n)}_{ab}-\mathring\nabla_{(a}\mathring\nabla^c\mu^{(n)}_{b)c}  = 
	-\left(\frac{n(n+1)}{2}-1\right)\Lambda \mu^{(n)}_{ab} \;.
\end{equation}
Using the same method for evaluating $\mathcal{R}^{(n)}_{ab}$ as for the second variation (see \eqref{eq_R_pert}) we find
\begin{equation}
	\mathcal{R}_{ab}^{(n)}  = 
	\mathring\Delta_L\mu^{(n)}_{ab} + \mathring\nabla_{(a}v^{(n)}_{b)}\;, \label{eq_R_npert}
\end{equation}
where $v^{(n)}_a=\mathring\nabla^b\mu^{(n)}_{ba}- \frac{1}{2} \mathring\nabla_a (\mathring\mu^{cd}\mu^{(n)}_{cd})$. 

Finally, using \eqref{eq_ntrace}, \eqref{eq_R_npert}, \eqref{eq_NHdata}, we  find that \eqref{eq_nRab} reduces to a simple eigenvalue equation 
\begin{equation}
	-\mathring\nabla^2\mu^{(n)}_{ab} = -\left(n^2 + n + \frac{2}{d-3}\right)\Lambda\mu^{(n)}_{ab}\;. \label{eq_ev_n}
\end{equation}
For $\Lambda>0$  the only solution is therefore $\mu^{(n)}_{ab}=0$.  Then substituting back into \eqref{eq_h_n} and \eqref{eq_F_n} we find that the rest of the $n$-th order data is
\begin{align}
\beta_a^{(n+1)}=0, \qquad 	\phi^{(n+2)} =  \frac{2\Lambda}{d-1}\left(-\frac{C}{2}\right)^n\frac{(d-2+n)!}{(d-3)!} \; .
\end{align}
Therefore, the claim follows by induction.\end{proof}

We can now deduce our main result which is a more detailed statement of Theorem \ref{thm_dS}.
\begin{theorem} \label{thm_main}
Consider an analytic spacetime that satisfies the Einstein equation \eqref{eq_Einstein} with $\Lambda>0$,   containing an extremal Killing horizon with a compact cross-section $S$ and near-horizon data \eqref{eq_NHdata}. Then the metric is given by \eqref{eq_metric_Gaussian} where 
\begin{align}
	\phi &=  \left\{
		\begin{array}{ll}
			-\frac{4}{C^2}\frac{\Lambda}{d-3}\left[1-\frac{2}{d-1}\left(1+\frac{C\rho}{2}\right)^{-d+3}-\frac{d-3}{d-1}\left(1+\frac{C\rho}{2}\right)^2\right] & \text{ for } C\neq0\;,\\
			\Lambda\rho^2&\text{ for } C=0\;,
		\end{array}\right.\nonumber\\
	\beta&=0\;,\label{eq_metric_C} \\
	\mu &=  r_0^2\left(1 +\frac{C\rho}{2}\right)^2d\Omega_{d-2}^2\; ,\nonumber
\end{align}
$C$ is a constant, and $r_0$ is defined in \eqref{eq_mmax}.  If $C=0$ this is the near-horizon geometry $dS_2\times S^{d-2}$ (Nariai solution). If $C\neq0$, this is the $d\geq 4$ extremal Schwarzschild de Sitter spacetime.
\end{theorem}

\begin{proof}
The assumption of analyticity means that metric components can be expressed as  Taylor series $\phi=\sum_{n\geq 0}\frac{ \phi^{(n)}}{n!} \rho^n$ and similarly for $\beta_a, \mu_{ab}$. The coefficients $\phi^{(n)}$ etc. are given by Propositions \ref{prop_1st}, \ref{prop_2nd} and \ref{prop_nth} and the resulting series can be summed to obtain \eqref{eq_metric_C}.  For $C=0$ one obtains $\td S_2\times S^{d-2}$ where the 2d de Sitter space $\td S_2$ is written in coordinates adapted to an extremal horizon. For $C\neq0$, one can 
rescale $\rho$ and $v$ by
\begin{align}
	v' = \frac{2}{ r_0 C}v\;, \qquad \rho' = \frac{ r_0 C}{2}\rho\;,\label{rescale}
\end{align}
so that the solution explicitly becomes the extremal Schwarzschild-dS solution \eqref{eq_metric_EF}, \eqref{eq_fr}.
\end{proof}

Finally, we consider higher order deformations in the $d=4$, $\Lambda<0$ case. In Proposition \ref{prop_AdS_1} we have seen  that there is a nontrivial space of first order transverse deformations of the horizon of extremal hyperbolic Schwarzschild-AdS, which are determined by the cohomology of the corresponding hyperbolic surface. It is an interesting question what solutions to the Einstein equation, if any, these cohomological first order deformations correspond to, but we will not pursue this here. 

Let us suppose that the first order deformation is that of the extremal Schwarzschild-AdS solution, that is, restrict to the $C_i=0$ deformations in Proposition  \ref{prop_AdS_1}. For such solutions, the proofs of Propositions \ref{prop_2nd} and \ref{prop_nth} show that non-trivial deformations starting at order $n\geq 2$ exist, if and only if nontrivial eigentensors of $-\mathring{\nabla}^2$ on $S$ with eigenvalue  $(n^2+n+2)| \Lambda |$ exist.  We can again expand the deformation $\mu_{ab}^{(n)}$ in the basis (\ref{eq_2tensors})  to deduce that nontrivial solutions starting  at  $n$-th order  exist, if and only if there exist eigenvalues of the scalar laplacian of the form 
\begin{equation}
\lambda_n = (n^2+n-2)|\Lambda |  \; ,\label{eq_eigenvalues_n}
\end{equation}
for $n\geq 2$. Therefore, for $d=4, \Lambda<0$ one is not guaranteed uniqueness at any order.    Nevertheless, we may state the following conditional result.

\begin{prop}\label{prop_AdS}
Consider a spacetime with an extremal horizon as in Proposition \ref{prop_AdS_1} and assume that the first order deformation is given by that of the extremal hyperbolic Schwarzschild-AdS solution (so $C_i=0$). Then Theorem \ref{thm_main} is  valid if the spectrum of the scalar laplacian on the hyperbolic surface $(S, \mathring \mu)$, where $\text{Ric}(\mathring\mu)=-|\Lambda | \mathring\mu$, does not contain any eigenvalue of the form \eqref{eq_eigenvalues_n} for  integer $n\geq 2$. In this case, the solution is either the near-horizon geometry $\text{AdS}_2\times H^2$ ($C=0$) or the extremal hyperbolic Schwarzschild-AdS spacetime ($C\neq 0$). 
\end{prop}

Interestingly, the spectrum of the laplacian on compact hyperbolic surfaces is an open problem and we are not 
aware of any analytic results that would in general rule out the eigenvalues \eqref{eq_eigenvalues_n}. For certain symmetric hyperbolic surfaces, including the Bolza surface and the Klein quartic, the first couple hundred eigenvalues have been computed numerically~\cite{strohmaier_algorithm_2011}, and none of them are of the form (\ref{eq_eigenvalues_n}) (in fact, none of them are integers in units of $|\Lambda |$). It is well-known that compact hyperbolic surfaces of genus $g$ have a $(6g-6)$-dimensional moduli space. It is natural to expect that, for a fixed genus, the spectrum is (generically) a  continuous function of the moduli. In fact, it has been shown that non-degenerate eigenvalues are analytic functions of the moduli~\cite{buser_geometry_2010}. Interestingly, numerical results suggest that for special points in the moduli space of hyperbolic surfaces, eigenvalues of the form (\ref{eq_eigenvalues_n}) can be realised~\cite{bonifacio_bootstrap_2021}. However, at generic points in the moduli space we expect that this is not the case. Therefore, it  seems reasonable to conjecture that, for generic compact hyperbolic surfaces, such eigenvalues do not occur and hence uniqueness of higher order deformations holds (assuming the first order deformation is as stated in Proposition \ref{prop_AdS}). On the other hand, if one drops the assumption that $S$ is compact, then eigenvalues of the scalar laplacian on hyperbolic space of the form \eqref{eq_eigenvalues_n} do always exist and therefore in this case uniqueness may be violated at all orders.
  \\

\noindent{\bf Acknowledgements.} DK is supported by an EPSRC studentship. JL is supported by a Leverhulme Trust Research Project Grant RPG-2019-355.  We thank James Bonifacio for useful comments on the spectrum of hyperbolic surfaces and pointing out the reference~\cite{bonifacio_bootstrap_2021}.

\section*{Statements and declarations}

\noindent {\bf Competing interests.}  The authors have no relevant financial or non-financial interests to disclose.
\\

\noindent{\bf Data availability.} Data sharing is not applicable to this article as no datasets were generated or analysed during the current study.

\appendix

\section{Ricci tensor in Gaussian null coordinates}\label{app_Ricci} 

The components of the Ricci tensor for a general metric in Gaussian null coordinates can be found in~\cite{moncrief_symmetries_1983} and 
also\footnote{We noticed that a $\frac{1}{2}\nabla^a\dot h_a$ term is missing in the $\rho v$ component 
in~\cite{hollands_stationary_2009}.} in~\cite{hollands_stationary_2009}.   The metric reads \eqref{eq_metric_Gaussian} where recall we also we assume $\partial_v$ is a Killing vector field.  Then, the components of the Ricci tensor of $g$ which we need are
\begin{align}
	R_{\rho\rho} &= -\frac{1}{2}\mu^{ab}\ddot \mu_{ab} + \frac{1}{4}\mu^{ac}\mu^{bd}\dot \mu_{ab}\dot \mu_{cd} \;,\\
	R_{\rho v} &= \frac{1}{2\sqrt{\det\mu}}\dotr{\left[\sqrt{\det \mu}\left(\dot{\phi}- 
	\beta^a\dot{\beta}_a\right)\right]}	+\frac{1}{2}\nabla^a\dot{\beta}_a\;,\\
	R_{\rho a} &= \frac{1}{2\sqrt{\det \mu}}\dotr{\left[\sqrt{\operatorname{det} \mu} \left(\dot{\beta}_a-\beta^b\dot \mu_{ab}\right)\right]}
	+\frac{1}{2}\nabla^b\dot \mu_{ab}-\frac{1}{2}\nabla_a(\mu^{bc}\dot\mu_{bc})\;, \\
	R_{ab} &=  \frac{1}{2\sqrt{\det\mu}}\dotr{\left[\sqrt{\det\mu}\left(2\nabla_{(a}\beta_{b)}
	+\phi\dot\mu_{ab}-\beta^c\beta_c\dot\mu_{ab}\right)\right]}+\frac{1}{2}\nabla_c\left(\beta^c\dot\mu_{ab}\right)
	\nonumber\\
	&\qquad+\mathcal{R}_{ab}-\frac{1}{2}\left[\dot\beta_a-\beta^c\dot\mu_{ac}\right]
	\left[\dot\beta_b-\beta^c\dot\mu_{bc}\right] -\dot\mu_{c(a}\nabla^c\beta_{b)} 
	+ \frac{1}{2}(\beta^c\beta_c-\phi)\dot\mu_{ac}\dot\mu_{bd}\mu^{cd}\; ,
\end{align}
where $\cdot$ denotes a $\rho$-derivative, $\nabla_a$ and $\mathcal{R}_{ab}$ are the metric connection and Ricci tensor of $\mu_{ab}$ on the codimension-2 surfaces $S_{v,\rho}$ of constant $(v,\rho)$. 

\section{Symmetric traceless \texorpdfstring{$2$}{2}-tensors on hyperbolic surfaces}\label{app_2tensor}

In this section we prove the following decomposition for traceless symmetric  tensor fields on a compact hyperbolic surface.

\begin{prop}
The space of traceless symmetric  $(0,2)$ tensor fields on a compact hyperbolic surface  $(S, \mathring\mu)$, with constant scalar curvature $R=2\Lambda<0$, is spanned by  eigentensors  of $-\mathring\nabla^2$,
\begin{align}
	S^{\{\lambda+4\}}_{ab} &:= \mathring{\nabla}_a\mathring{\nabla}_b f^{\{\lambda\}} - \frac{1}{2}\mathring \mu_{ab}\mathring\nabla^2f^{\{\lambda\}}\;, \qquad
	P^{\{\lambda+4\}}_{ab} := \mathring{\nabla}_{(a}\mathring\epsilon_{b)c}\mathring\nabla^cf^{\{\lambda\}}\;, \nonumber\\
	V^{\{4\}, i}_{ab} &:= \mathring{\nabla}_{(a}\xi^i_{b)}\;, \qquad\qquad\qquad\qquad\qquad Y^{\{2\}}_{ab}\;,\label{eq_2tensors_app}
\end{align}
where $\mathring\epsilon$ is the volume form, $f^{\{\lambda\}}$ are eigenfunctions of $-\mathring\nabla^2$ with 
eigenvalues $-\lambda \Lambda > 0$, $\{\xi^i\}_{i=1}^{2g}$ is a basis of harmonic 1-forms with respect to the Hodge--de Rham laplacian, and $Y^{\{2\}}$ is a  traceless symmetric divergence-free $(0,2)$ tensor.  In particular,  $S^{\{\lambda+4\}}, P^{\{\lambda+4\}}, V^{\{4\}}, Y^{\{2\}}$ are eigentensors of $-\mathring\nabla^2$ with eigenvalues $-(\lambda+4)\Lambda $, $-(\lambda+4)\Lambda $, $-4 \Lambda $ and $-2 \Lambda $, respectively.
\end{prop}

\begin{proof}
 For any traceless symmetric  tensor $T_{ab}$ we claim that there exists a 1-form $X_a$ such that
\begin{equation}
	Y_{ab}:= T_{ab}-\mathring\nabla_{(a}X_{b)}+\frac{1}{2}\mathring\mu_{ab}\mathring\nabla^c X_c \label{eq_Ydef}
\end{equation}
is divergence-free (note that it is also trace-free by construction). Taking the divergence of (\ref{eq_Ydef}) and requiring that $\mathring\nabla^b Y_{ab}=0$ we obtain that such $X$ must satisfy
\begin{equation}
	\left(-\mathring\nabla^2 - \Lambda  \right)X_a = -2\mathring\nabla^bT_{ab}\;.
\end{equation}
This equation always has a unique solution for $X$ and compact $S$, since the operator on the left-hand side is elliptic and self-adjoint and has a trivial kernel (see e.g. Theorem 5.22 in \cite{voisin_hodge_2002}).  This establishes the claim of existence of the 1-form $X$ above.

Next,  by the Hodge-decomposition theorem $X=\td f+\star\td g+\xi$, where $f, g$ are functions on $S$, $\star$ is the Hodge operator on $S$, and $\xi$ is a harmonic 1-form with respect to the Hodge--de Rham laplacian $(\td + \star\td\star)^2=-\mathring\nabla^2+\Lambda$ (for 1-forms on $S$). It follows that we can always decompose a traceless symmetric tensor $T_{ab}$ in terms of functions $f, g$, a harmonic 1-form $\xi$, and a divergence-free traceless symmetric tensor $Y_{ab}$, as 
\begin{equation}
T_{ab}=  \mathring{\nabla}_a\mathring{\nabla}_b f - \frac{1}{2}\mathring \mu_{ab}\mathring\nabla^2 f  +  \mathring{\nabla}_{(a}\mathring\epsilon_{b)c}\mathring\nabla^c g   +\nabla_{(a} \xi_{b)} + Y_{ab}  \; .
\end{equation}
The functions $f, g$ can be each expanded in a basis of eigenfunctions $f^{ \{\lambda \} }$ of $-\mathring\nabla^2$. The space of harmonic 1-forms on a hyperbolic surface of genus $g\geq 2$  corresponds to the first cohomology which is $2g$-dimensional. Thus, using the above decomposition, we deduce that  \eqref{eq_2tensors_app} span the space of traceless symmetric tensor fields.  

Finally, the claim that \eqref{eq_2tensors_app} are all eigentensors of $-\mathring\nabla^2$ with the stated eigenvalues follows by explicit calculation using the fact that $(S, \mathring\mu)$ is maximally symmetric.
\end{proof}

Interestingly, the above decomposition for traceless symmetric tensors on hyperbolic surfaces is more complicated than for the sphere $S^2$.   As is well known, on $S^2$ the scalar-derived tensor harmonics (i.e. $S^{\{\lambda+4\}}, P^{\{\lambda+4\}}$)  span the space of traceless symmetric tensor fields.  This can be shown using similar arguments as above, noting that there are no harmonic 1-forms or divergence-free traceless symmetric tensors on $S^2$.

\bibliographystyle{sn-mathphys_modified}
\bibliography{ref}
\end{document}